\UseRawInputEncoding
\documentclass[%
 reprint,
superscriptaddress,
 amsmath,amssymb,
 aps,
pra,
]{revtex4-1}

\usepackage{amsthm,amssymb}
\usepackage{amsmath}
\usepackage{graphicx}
\usepackage{dcolumn}
\usepackage{bm}
\usepackage{multirow} 
\usepackage{mathdots}
\usepackage{enumerate}
\usepackage{algorithm}
\usepackage{algorithmic}
\usepackage{color}
\usepackage{epsfig}
\usepackage{subfigure}
\usepackage{caption}

\theoremstyle{remark}

\newtheorem{theorem}{\indent \emph{\textbf{Theorem}}}
\newtheorem{corollary}{\indent \textbf{\emph{Corollary}}}

\begin{document}

\makeatletter
\newcommand{\ud}{\mathrm{d}}
\newcommand{\rmnum}[1]{\romannumeral #1}
\newcommand{\Rmnum}[1]{\expandafter\@slowromancap\romannumeral #1@}
\newcommand{\udots}{\mathinner{\mskip1mu\raise1pt\vbox{\kern7pt\hbox{.}}
        \mskip2mu\raise4pt\hbox{.}\mskip2mu\raise7pt\hbox{.}\mskip1mu}}
\makeatother

\preprint{APS/123-QED}

\title{An improved quantum algorithm for A-optimal projection}
\author{Shi-Jie Pan}
\affiliation{State Key Laboratory of Networking and Switching Technology, Beijing University of Posts and Telecommunications, Beijing, 100876, China}
\affiliation{State Key Laboratory of Cryptology, P.O. Box 5159, Beijing, 100878, China}
\affiliation{Center for Quantum Computing, Peng Cheng Laboratory, Shenzhen 518055, China}
\author{Lin-Chun Wan}
\affiliation{State Key Laboratory of Networking and Switching Technology, Beijing University of Posts and Telecommunications, Beijing, 100876, China}
\author{Hai-Ling Liu}
\affiliation{State Key Laboratory of Networking and Switching Technology, Beijing University of Posts and Telecommunications, Beijing, 100876, China}
\author{Qing-Le Wang}
\affiliation{School of Control and Computer Engineering, North China Electric Power University, Beijing, 102206, China}
\affiliation{CAS Key Laboratory of Quantum Information, University of Science and Technology of China, Hefei 230026, China}
\author{Su-Juan Qin}
\email{qsujuan@bupt.edu.cn}
\affiliation{State Key Laboratory of Networking and Switching Technology, Beijing University of Posts and Telecommunications, Beijing, 100876, China}
\author{Qiao-Yan Wen}
\email{wqy@bupt.edu.cn}
\affiliation{State Key Laboratory of Networking and Switching Technology, Beijing University of Posts and Telecommunications, Beijing, 100876, China}
\author{Fei Gao}
\email{gaof@bupt.edu.cn}
\affiliation{State Key Laboratory of Networking and Switching Technology, Beijing University of Posts and Telecommunications, Beijing, 100876, China}
\affiliation{Center for Quantum Computing, Peng Cheng Laboratory, Shenzhen 518055, China}

\date{\today}

\begin{abstract}
Dimensionality reduction (DR) algorithms, which reduce the dimensionality of a given data set while preserving the information of the original data set as well as possible, play an important role in machine learning and data mining. Duan \emph{et al}. proposed a quantum version of the A-optimal projection algorithm (AOP) for dimensionality reduction [Phys. Rev. A 99, 032311 (2019)] and claimed that the algorithm has exponential speedups on the dimensionality of the original feature space $n$ and the dimensionality of the reduced feature space $k$ over the classical algorithm. In this paper, we correct the time complexity of Duan \emph{et al}.'s algorithm to $O(\frac{\kappa^{4s}\sqrt{k^s}} {\epsilon^{s}}\mathrm{polylog}^s (\frac{mn}{\epsilon}))$, where $\kappa$ is the condition number of a matrix that related to the original data set, $s$ is the number of iterations, $m$ is the number of data points and $\epsilon$ is the desired precision of the output state. Since the time complexity has an exponential dependence on $s$, the quantum algorithm can only be beneficial for high dimensional problems with a small number of iterations $s$. To get a further speedup, we propose an improved quantum AOP algorithm with time complexity $O(\frac{s \kappa^6 \sqrt{k}}{\epsilon}\mathrm{polylog}(\frac{nm}{\epsilon}) + \frac{s^2 \kappa^4}{\epsilon}\mathrm{polylog}(\frac{\kappa k}{\epsilon}))$ and space complexity $O(\log_2(nk/\epsilon)+s)$. With space complexity slightly worse, our algorithm achieves at least a polynomial speedup compared to Duan \emph{et al}.'s algorithm. Also, our algorithm shows exponential speedups in $n$ and $m$ compared with the classical algorithm when both $\kappa$, $k$ and $1/\epsilon$ are $O(\mathrm{polylog}(nm))$.
\end{abstract}

\pacs{Valid PACS appear here}
\maketitle

\section{Introduction}
Quantum computing is more computationally powerful than classical computing in solving specific problems, such as the factoring problem \cite{PS1994}, unstructured data search problem \cite{GL1996} and matrix computation problems \cite{HHL, Wan2018}. In recent years, quantum machine learning (QML) has received wide attention as an emerging research area that successfully combines quantum physics and machine learning. An important part of the study of QML focuses on designing quantum algorithms to speed up the machine learning problems, such as data classification \cite{LMR, RML, CD, SFP, DYLL}, linear regression \cite{WBL, SSP, W, YGW, YGLHRW}, association rules mining \cite{YGWW} and anomaly detection \cite{LR}.

In the big data era, most of the real-world data are high-dimensional, which requires high computational performance and usually causes a problem called \emph{curse of dimensionality} \cite{Bpattern}. Since the high-dimensional real-world data are often confined to a region of the space having lower effective dimensionality \cite{Bpattern}, a technique called dimensionality reduction (DR) which reduces the dimensionality of the given data set while preserving the information of the original data set as well as possible was proposed. The DR algorithm often serves as a preprocessing step in data mining and machine learning.

Based on the feature space that the data lie on and the learning task that we want to handle, various DR algorithms have been developed. Generally, when the data lie on a linear embedded manifold, principal component analysis (PCA) \cite{HHPCA}, a DR algorithm maintaining the characteristics of the data set that contribute the most to the variance, is guaranteed to uncover the intrinsic dimensionality of the manifold. When the data lie on a non-linearly embedded manifold, the manifold learning techniques, such as Isomap \cite{TSL2000}, Locally Linear Embedding \cite{Roweis2323}, and Laplacian Eigenmap \cite{BP2002} can be used to discover the nonlinear structure of the manifold.Since the DR algorithms mentioned above aim to discover the geometrical or cluster structure of the training data, these algorithms are not directly related to the classification and regression tasks which are the two most important tasks in machine learning and data mining. For the classification task, a famous DR technique called linear discriminant analysis (LDA) was put forward, which maximizes the ratio of the
between-class variance and the within-class variance of the training data [22] \cite{Fisher1936}. For the regression task, He \emph{et al}. proposed a novel DR algorithm called A-Optimal Projection (AOP) that aims to minimize the prediction error of a regression model while reducing the dimensionality \cite{HZCL}. Their algorithm improves the regression performance in the reduced space.

In the context of quantum computing, the quantum PCA was proposed by Lloyd \emph{et al}. to reveal in quantum form
the eigenvectors corresponding to the large eigenvalues of an unknown low-rank density matrix \cite{Lloyd_2014}. Later, Yu \emph{et al}. proposed a quantum algorithm that compresses training data based on PCA \cite{YGLW}. When the dimensionality of the reduced space is polylogarithmic in the training data, their quantum algorithm achieves an exponential speedup compared with the classical algorithm. Cong \emph{et al}. implemented a quantum LDA algorithm which has an exponential speedup in the scales of the original data set compared with the classical algorithm \cite{CD}. In \cite{DYXL}, Duan \emph{et al}. studied the AOP algorithm and proposed its quantum counterpart, called Duan-Yuan-Xu-Li (DYXL) algorithm. The DYXL algorithm is iterative and was expected to have a time complexity $O(s\mathrm{polylog}(nk/\epsilon)$, where $s$ is the number of iterations, $n$ is the dimensionality of the original feature space, $k$ is the dimensionality of the reduced feature space and $\epsilon$ is the desired precision of the output state.

In this paper, we reanalyze the DYXL algorithm and correct the time complexity to $O(\frac{\kappa^{4s}\sqrt{k^s}} {\epsilon^{s}}\mathrm{polylog}^s (\frac{mn}{\epsilon}))$, where $\kappa$ is the condition number of a matrix that related to the original data set, $m$ is the number of data points. We find that in the DYXL algorithm, multiple copies of the current candidate are consumed to improve the candidate by quantum phase estimation and post-selection in each iteration, which results in the total time complexity having exponential dependence on the number of iterations $s$. Thus the DYXL algorithm can only be beneficial for high dimensional problems with a small $s$, which limits the practical application of the algorithm.
To get a further speedup and reduce the dependence on $s$, we propose an improved quantum AOP algorithm with time complexity $O(\frac{s \kappa^6 \sqrt{k}}{\epsilon}\mathrm{polylog}(\frac{nm}{\epsilon}) + \frac{s^2 \kappa^4}{\epsilon}\mathrm{polylog}(\frac{\kappa k}{\epsilon}))$. Note that in the DYXL algorithm, one only changes the amplitude of the candidate in each iteration. In our algorithm, we process the amplitude information of the candidate in computational basis to reduce the consumption of the copies of the current candidate. Our algorithm has a quadratic dependence rather than an exponential dependence on $s$ in time complexity, which achieves a significant speedup over the DYXL algorithm with the space complexity slightly worse. Also, it shows exponential speedups over the classical algorithm on $n$ and $m$, when both $\kappa$, $k$ and $1/\epsilon$ are $O(\mathrm{polylog}(nm))$.

The rest of this paper is organized as follows. In Sec. \Rmnum{2}, we review the classical AOP algorithm in Sec. \Rmnum{2} A and its quantum version in Sec. \Rmnum{2} B, and analyze the complexity of the DYXL algorithm in Sec. \Rmnum{2} C. We then propose our quantum AOP algorithm and analyze the complexity in Sec. \Rmnum{3}. In Sec. \Rmnum{4}, we discuss the number of iterations of the two quantum algorithms. The conclusion is given in Sec. \Rmnum{5}.

\section{Review of the classical and quantum AOP algorithm}
\label{sec:2}
In this section, we will briefly review the AOP algorithm in Sec. \Rmnum{2} A. The DYXL algorithm will be introduced in Sec. \Rmnum{2} B, and we will analyze its complexity in Sec. \Rmnum{2} C.

\subsection{Review of AOP algorithm}
Suppose $X=(\textbf{x}_1, \textbf{x}_2,..., \textbf{x}_m)$ is a data matrix with dimension $n \times m$, where $n$ is the number of the features and $m$ is the number of data points. The objective of the AOP is to find the optimal projection matrix $A \in n \times k$ which minimizes the trace of the covariance matrix of regression parameters to reduce the dimensionality of $X$.

In He \emph{et al}. \cite{HZCL}, a graph regularized regression model was chosen and thus the optimal projection matrix $A$ can be obtained by solving the following objective function:
    \begin{eqnarray}\label{eq:equ1}
            \mathop{\min}_A \mathop{Tr} \left[ \left( A^T X \left( I+ \lambda_1 L \right)X^T A + \lambda_2 I \right)^{-1} \right],
    \end{eqnarray}
where $\lambda_1$ and $\lambda_2$ are the regularized coefficients, $L=\mathrm{diag}\left(S\bf{1}\right)-S$ is  \emph{graph Laplacian} where $S$ is the weight matrix of the data points and $\bf{1}$ is a vector of all ones. Let $\mathcal{N}_k(\mathbf{x})$ denotes the $k$ nearest neighbors of $\mathbf{x}$, a simple definition of $S$ is as follows:
    \begin{eqnarray}\label{eq:equ2}
         S_{i,j} =
             \begin{cases}
                1, & \mbox{if }{{\bf{x}}_i}\in \mathcal{N}_k\left({{\bf{x}}_j}\right)\mbox{ }\mathrm{or} \mbox{  } {{\bf{x}}_j} \in \mathcal{N}_k\left({\bf{x}}_i\right); \\
                0, & \mbox{otherwise}.
             \end{cases}
    \end{eqnarray}
 To solve the optimization problem (\ref{eq:equ1}), He \emph{et al}. introduced a variable $B$ and proposed the following theorem \cite{HZCL}:
\begin{theorem} (\emph{Theorem} \emph{4.3} in \cite{HZCL})
The optimization problem (\ref{eq:equ1}) is equivalent to the following optimization problem:
    \begin{eqnarray}\label{eq:equ3}
        \mathop{\min}_{A,B}{\left\| {I-{A^T} {\widetilde{X}} {B}} \right\|}^2+\lambda{\left\|{B}\right\|}^2,
    \end{eqnarray}
where $\widetilde{X}=X\Sigma$ and $\Sigma$ is defined by the equation $I+\lambda_1 L=\Sigma\Sigma^T$.
\end{theorem}
Then we can use the iterative method to find the optimal $A$. The procedure of computing the projection matrix $A$ can be summarized as follows:
\begin{enumerate}
    \item Initialize the matrix $A$ by computing the PCA of the data matrix $X$;
    \item Computing matrix $B$ according to equation (\ref{eq:equ4}):
    \begin{eqnarray}\label{eq:equ4}
        B=\left({\widetilde{X}^T A A^T \widetilde{X}+\lambda_2 I}\right)^{-1} \widetilde{X}^T A.
    \end{eqnarray}
    \item Computing matrix $A$ according to equation (\ref{eq:equ5}):
    \begin{eqnarray}\label{eq:equ5}
        A=\left({\widetilde{X} B B^T \widetilde{X}^T}\right)^{-1} \widetilde{X} B.
    \end{eqnarray}
    Normalize $A$ to satisfy $\left\|{A}\right\|_F\le\rho$ ($\rho$ is a constant and here we set it to 1).
    \item Repeat steps 2 and 3 until convergence.
\end{enumerate}

Since the AOP algorithm involves matrix multiplication and inversion, the time complexity of the classical algorithm is $\Omega(s\mathrm{poly}(nm))$, where $s$ is the number of iterations.

\subsection{Review of the DYXL algorithm}
\label{sub:2b}
In \cite{DYXL}, the authors reformulated the iterative method of AOP to make the algorithm suitable for quantum settings. They adjusted the initialization of matrix $A$ and combined the steps 2 and 3 into one
step to remove the variable B.
Suppose the singular value decomposition of the matrix $\widetilde{X}$ is $\widetilde{X}=\sum_{j=0}^{r-1} \sigma_j |\textbf{u}_j\rangle\langle \textbf{v}_j|$, where $r =O(\mathrm{polylog}(m,n))$ is the rank of $\widetilde{X}$, $|\textbf{u}_j\rangle$ and $|\textbf{v}_j\rangle$ are the left and right singular vectors with corresponding singular value $\sigma_j (1 = \sigma_0 \ge \sigma_1\ge ... \ge \sigma_{r-1} > 0)$. The reformulated AOP algorithm can be summarized as follows:
\begin{enumerate}
    \item Initialize the matrix $A^{(0)}$ by computing the PCA of the data matrix $\widetilde{X}$,
    \begin{eqnarray}\label{eq:equ6}
        A^{(0)}=\mathrm{PCA}{(\widetilde{X})}=\sum_{j=0}^{k-1} |\textbf{u}_j\rangle\langle \textbf{j}|,
    \end{eqnarray}
where $A^{(i)}$ is the matrix $A$ of iteration $i$, $k$ is the rank of $A^{(0)}$ and $|\textbf{j}\rangle$ is the computational basis state.
    \item Update the matrix $A$ according to the following equation
    \begin{eqnarray}\label{eq:equ7}
        A^{(i)}=\sum_{j=0}^{k-1} \beta_j^{(i)} |\textbf{u}_j\rangle\langle \textbf{j}|=\sum_{j=0}^{k-1} \frac{(\sigma_j \beta_j^{(i-1)})^2+\lambda_2}{c^{(i)}\sigma_j^{2} \beta_j^{(i-1)}} |\textbf{u}_j\rangle\langle \textbf{j}|,
    \end{eqnarray}
    where $\beta_j^{(i)}$ is the singular value of $A^{(i)}$ with corresponding left and right singular vectors $|\textbf{u}_j\rangle$ and $|\textbf{j}\rangle$, $c^{(i)}$ is a constant to ensure that $\left\|{A^{(i)}}\right\|_F\le 1$.

    \item Repeat step 2 until convergence.
\end{enumerate}

The DYXL algorithm can be summarized as follows:
\begin{enumerate}
    \item Initialize $i=0$, and prepare the state $|\psi_{A^{(0)}}\rangle$, where
    \begin{eqnarray}\label{eq:equ8}
        |\psi_{A^{(0)}}\rangle=\frac{1}{\sqrt{k}}\sum_{j=0}^{k-1}|\textbf{u}_j\rangle |\textbf{j}\rangle.
    \end{eqnarray}

    \item Suppose the quantum state $|\psi_{A^{(i-1)}}\rangle$ is given, prepare the following state:
    \begin{eqnarray}\label{eq:equ9}
    \begin{split}
        &|\psi_{0}^{(i-1)}\rangle =|0\rangle^{D} (|0\rangle ... |0\rangle)^C(|0\rangle ... |0\rangle)^B |\psi_{A^{(i-1)}}\rangle^A \\
        &=|0\rangle^{D}\sum_{j=0}^{k-1} \beta_j^{(i-1)} (|0\rangle ... |0\rangle)^C(|0\rangle ... |0\rangle)^B (|\textbf{u}_j\rangle|\textbf{j}\rangle)^A,
    \end{split}
    \end{eqnarray}
    where the superscripts $D$, $C$, $B$, $A$ represent the register $D$, $C$, $B$, and $A$, respectively.

    \item Perform phase estimation on the $|\psi_{0}^{(i-1)}\rangle$ for the unitary $e^{i\widetilde{X}\widetilde{X}^\dagger t_0}$ and $e^{iA^{(i-1)}{{A^{(i-1)}}^{\dagger} t_0}}$,
    \begin{eqnarray}\label{eq:equ10}
    \begin{split}
        &|\psi_{1}^{(i-1)}\rangle=  \\
        &|0\rangle^{D} \!\sum_{j=0}^{k-1} \beta_j^{(i-1)} |\sigma_j^2\rangle^C |(\beta_j^{(i-1)})^2\rangle^B\!(|\textbf{u}_j\rangle|\textbf{j}\rangle)^A.
    \end{split}
    \end{eqnarray}

     \item Perform an appropriate controlled rotation
     on the register $B$, $C$ and $D$, transforms the system to:
     \begin{eqnarray}\label{eq:equ11}
     \begin{split}
        |\psi_{2}^{(i-1)}\rangle=&\sum_{j=0}^{k-1} \beta_j^{(i-1)} |\sigma_j^2\rangle^C |(\beta_j^{(i-1)})^2\rangle^B(|\textbf{u}_j\rangle|\textbf{j}\rangle)^A \\
        &(\sqrt{1-\rho^2 f(\sigma_j,\beta_j^{(i-1)})^2}|0\rangle+\rho f(\sigma_j,\beta_j^{(i-1)})|1\rangle)^{D},
     \end{split}
     \end{eqnarray}
     where $\rho$ is a constant to ensure $|\rho f(\sigma_j,\beta_j^{(i-1)})| \le 1,$ $ f(\sigma_j,\beta_j^{(i-1)})=\frac{(\sigma_j \beta_j^{(i-1)})^2+\lambda_2}{ (\sigma_j {\beta_j^{(i-1)}})^2}$.
     \item Measure the register $D$, then uncompute the register $C$, $B$ and $A$, and remove the register $C$, $B$. Conditioned on seeing 1 in $D$, we have the state
    \begin{eqnarray}\label{eq:equ12}
      \begin{split}
        |\psi_{3}^{(i-1)}\rangle &=\frac{1}{\sqrt{N^{(i)}}}\sum_{j=0}^{k-1} \frac{(\sigma_j \beta_j^{(i-1)})^2+\lambda_2}{\sigma_j^{2} \beta_j^{(i-1)}} |\textbf{u}_j\rangle |\textbf{j}\rangle\\
        &=\sum_{j=0}^{k-1} \beta_j^{(i)} |\textbf{u}_j\rangle|\textbf{j}\rangle
        =|\psi_{A^{(i)}}\rangle,
      \end{split}
    \end{eqnarray}
    where $N^{(i)}=\sum_{j=0}^{k-1} (\frac{(\sigma_j \beta_j^{(i-1)})^2+\lambda_2}{\sigma_j^{2} \beta_j^{(i-1)}})^2$. Thus $\beta_j^{(i)} \in [0,1]$ for $j \in \{0,1,2,...,k-1\}$ and $i \ge 0$.
    \item For $i=1$ to $s-1$, repeat step 2 to 5.
\end{enumerate}

\subsection{Complexity analysis of the DYXL algorithm}

In \cite{DYXL}, the authors analyzed the time complexity of each iteration (step 2 to step 5 in this paper) and claimed that the total time complexity is the product of the number of iterations and the time complexity of each iteration. Actually, in the $i$th iteration, the algorithm has to prepare the state $|\psi_{A^{(i-1)}}\rangle$ several times to perform phase estimation in step 3 and do measurements to obtain an appropriate state in step 5, which means that the total time complexity is exponential on the number of iterations $s$. The time complexity of each step can be seen in TABLE \ref{tab:stc} and the proof details can be seen in appendix \ref{app:complexity}.

\begin{table}[!htb]
\caption{\label{tab:stc}
The time complexity of each step of the DYXL algorithm.}
\begin{ruledtabular}
\begin{tabular}{cc}
Steps &Time complexity \\ \hline
Step 1 &$O(\log_2 (\epsilon^{-1}) \log_2(nk))$ \\
Step 3 &$O((G^{(i-1)}/\epsilon_1) \log_2 (1/\epsilon_1))+$ \\
       &$O\left((1/\epsilon_1) \mathrm {polylog}(nm/\epsilon_1)\right)$ \\
Step 4 & $O(\mathrm{polylog} (1/\epsilon))$ \\
Step 5 & $O(\kappa^2)$ repetitions \\
\end{tabular}
\end{ruledtabular}
Here the step 3-5 is the steps of the $i$th iteration and we neglect the runtime of step 2. $G^{(i-1)}$ is the time complexity to prepare state $|\psi_{A^{(i-1)}}\rangle$, $\kappa$ is the condition number of $\widetilde{X}$, $\epsilon$ is the desired precision of the output state, $\epsilon_1= O(\frac{\epsilon}{\kappa^2\sqrt{k}})$.
\end{table}

Putting all together, the runtime of the $i$th iteration (i.e., preparing the state $|\psi_{A^{(i)}}\rangle$) is
\begin{eqnarray*}
\begin{split}
G^{(i)}\!
&= \! O\left( \!(\frac{G^{(i-1)}}{\epsilon_1} \log_2 (\frac{1}{\epsilon_1})+ \frac{1}{\epsilon_1} \mathrm {polylog}\frac{mn}{\epsilon_1}+ \mathrm{polylog} (\frac{1}{\epsilon}))\kappa^2\!\right) \\
&= O\left(\frac {\kappa^4\sqrt{k} G^{(i-1)}}{\epsilon} \mathrm{polylog}(\frac{mn} {\epsilon}) \right)=O(TG^{(i-1)}),
\end{split}
\end{eqnarray*}
where $T=(\kappa^4\sqrt{k}/\epsilon)\mathrm{polylog}(mn/ \epsilon)$. Since $G^{(0)}=O(\log_2 (\epsilon^{-1}) \log_2 (nk))$, the overall time complexity of the algorithm is
\begin{eqnarray}
\begin{split}
G^{(s)}
& = O(TG^{(s-1)})
  =O(T^{s} G^{(0)})  \\
& = O(\frac{\kappa^{4s}\sqrt{k^s}}{\epsilon^{s}}\mathrm{polylog}^s (mn/ \epsilon)).
\end{split}
\end{eqnarray}

As for the space complexity, $O(\mathrm{log_2}(nk/\epsilon))$ qubits are used to prepare the initial state $|\psi_{A^{(0)}}\rangle$. In step 2 and step 3, the quantum phase estimations require $O(\mathrm{log_2}(1/\epsilon_1))$ qubits. In step 4, the controlled rotation requires $O(\mathrm{log_2}(1/\epsilon))$ ancillary qubits. Note that the qubits in the current iteration can be reused in the next iteration, thus the space complexity of the algorithm is $O(\mathrm{log_2}(nk/\epsilon))$. The details are shown in the Appendix \ref{app:complexity}.

\section{An improved quantum AOP algorithm}
In this section, we present an improved quantum AOP algorithm.

In the DYXL algorithm, to get the state $|\psi_{A^{(i)}}\rangle$ from state $|\psi_{A^{(i-1)}}\rangle$, one performs phase estimation and post selection, which consumes multiple copies of $|\psi_{A^{(i-1)}}\rangle$, thus the number of copies of the initial state $|\psi_{A^{(0)}}\rangle$ required depends exponentially on the number of iterations $s$. Note that in each iteration of the DYXL algorithm, one only changes the eigenvalue $\beta_j$ for $ j=0,1,..,k-1$. In our algorithm, we put the calculation into the computational basis to reduce the consumption of the copies of $|\psi_{A^{(i-1)}}\rangle$.

\subsection{An improved quantum AOP algorithm}

The specific process of our quantum algorithm is as follows.

1. \emph{Initialization} Initialize $i=0$, and prepare the state $|\psi_{A^{(0)}}\rangle=\frac{1}{\sqrt{k}}\sum_{j=0}^{k-1}|\textbf{u}_j\rangle |\textbf{j}\rangle.$

2. \emph{Prepare state $|\psi_0\rangle$} Perform phase estimation with precision parameter $\epsilon$ on the state $|\psi_{A^{(0)}}\rangle$ for the unitary $e^{i\widetilde{X} \widetilde{X}^ \dagger t_0}$, and then append state $|\frac{1}{\sqrt{k}}\rangle|0\rangle$, thus we obtain
    \begin{eqnarray}
    \begin{split}
        |\psi_{0}\rangle
        &= \frac{1}{\sqrt{k}}\sum_{j=0}^{k-1}(|\textbf{u}_j\rangle |\textbf{j}\rangle)^A |\sigma_j^2 \rangle^B |\frac{1}{\sqrt{k}} \rangle^C |0\rangle^D \\
        &= \frac{1}{\sqrt{k}}\sum_{j=0}^{k-1}(|\textbf{u}_j\rangle |\textbf{j}\rangle)^A |\sigma_j^2 \rangle^B |\beta_j^{(0)}\rangle^C |0\rangle^D, \\
    \end{split}
    \end{eqnarray}
where $\beta_j^{(0)}=\frac{1}{\sqrt{k}},$ for $j=0,1,...,k-1$, the superscript $A$, $B$, $C$, $D$ represent the register $A$, $B$, $C$ and $D$, respectively (in the absence of ambiguity, we omit these superscripts below for the sake of simplicity).

Assuming that we can prepare the state $|\psi_{i-1}\rangle$, where
    \begin{eqnarray}
    \begin{aligned}
        |\psi_{i-1}\rangle =\frac{1}{\sqrt{k}}\sum_{j=0}^{k-1}|\textbf{u}_j\rangle |\textbf{j}\rangle|\sigma_j^2 \rangle |\beta_j^{(i-1)}\rangle |0 \rangle.
    \end{aligned}
    \end{eqnarray}
Thus we could perform quantum arithmetic operations to get
    \begin{eqnarray}
    \begin{split}\label{eq:19}
        |\phi_1^{(i)}\rangle =\frac{1}{\sqrt{k}}\sum_{j=0}^{k-1}|\textbf{u}_j\rangle |\textbf{j}\rangle|\sigma_j^2 \rangle |\beta_j^{(i-1)}\rangle |c^{(i)}\beta_j^{(i)} \rangle,\\
    \end{split}
    \end{eqnarray}
where $c^{(i)} \beta_j^{(i)} = \frac{(\sigma_j \beta_j^{(i-1)})^2+ \lambda_2}{\sigma_j^2 \beta_j^{(i-1)}}$ and $\sum_{j=0}^{k-1} (\beta_j^{(i)})^2=1$.

In order to obtain the information of $\beta_j^{(i)}$, we will estimate $c^{(i)}$ first, then we can prepare the state $\frac{1} {\sqrt{k}} \sum_{j=0}^{k-1}|\textbf{u}_j\rangle |\textbf{j}\rangle|\sigma_j^2 \rangle |\beta_j^{(i-1)}\rangle |\beta_j^{(i)} \rangle:=|\phi_3^{(i)}\rangle$ from the state $|\phi_1^{(i)}\rangle$.

3. \emph{Estimate $c^{(i)}$} Assuming that we can prepare the state $|\psi_{i-1}\rangle$ in time $G_{i-1}$.

(\romannumeral1) Prepare the state $|\phi_1^{(i)}\rangle$ from the state $|\psi_{i-1}\rangle$.

(\romannumeral2) Add an ancillary qubit (register $E$) and perform an appropriate controlled rotation on the registers $D$ and $E$, transforms the system to:
     \begin{eqnarray}
     \begin{split}
        |\phi_2^{(i)}\rangle=&\frac{1}{\sqrt{k}}\sum_{j=0}^{k-1}(|\textbf{u}_j\rangle |\textbf{j}\rangle)^A|\sigma_j^2 \rangle^B |\beta_j^{(i-1)}\rangle^C |c^{(i)}\beta_j^{(i)} \rangle^D \\
        &(\sqrt{1-(\frac{c^{(i)}\beta_j^{(i)}}{c})^2}|0\rangle+\frac{c^{(i)}\beta_j^{(i)}}{c}|1\rangle)^E \\
        &:=\cos(\theta)|a\rangle|0\rangle^E + \sin(\theta) |b\rangle |1\rangle^E,
     \end{split}
     \end{eqnarray}
where the parameter $c$ is a constant to ensure $\frac{c^{(i)}\beta_j^{(i)}}{c} \le 1$,
\begin{gather}
|a\rangle= \sum_{j=0}^{k-1}\sqrt{\frac{c^2-(c^{(i)}\beta_j^{(i)})^2}{kc^2-(c^{(i)})^2}}|\textbf{u}_j\rangle |\textbf{j}\rangle|\sigma_j^2 \rangle |\beta_j^{(i-1)}\rangle |c^{(i)}\beta_j^{(i)} \rangle,\notag\\
|b\rangle= \sum_{j=0}^{k-1}\beta_j^{(i)}|\textbf{u}_j\rangle |\textbf{j}\rangle|\sigma_j^2 \rangle |\beta_j^{(i-1)}\rangle |c^{(i)}\beta_j^{(i)} \rangle,\notag\\
\label{eq:21}\sin(\theta)= \frac{c^{(i)}}{c\sqrt{k}}.
\end{gather}

(\romannumeral3) Perform quantum amplitude estimation to estimate $\sin(\theta)$. Then we can obtain the classical information of $c^{(i)}$ by $c^{(i)}=\sqrt{k}c\sin(\theta)$.

4. \emph{Prepare state $|\psi_{i}\rangle$}

(\romannumeral1) Since we have the classical information of $c^{(i)}$, we can perform quantum arithmetic operation to get
    \begin{eqnarray}
    \begin{split}\label{eq:22}
        |\phi_3^{(i)}\rangle =\frac{1}{\sqrt{k}}\sum_{j=0}^{k-1}|\textbf{u}_j\rangle |\textbf{j}\rangle|\sigma_j^2 \rangle^B |\beta_j^{(i-1)}\rangle^C |\beta_j^{(i)} \rangle^D.
    \end{split}
    \end{eqnarray}

Note that by using the techniques from step 2 to stage (\romannumeral1) of step 4, we could prepare the state $\frac{1}{\sqrt{k}}\sum_{j=0}^{k-1}|\textbf{u}_j\rangle |\textbf{j}\rangle|\sigma_j^2 \rangle |\beta_j^{(0)}\rangle|\beta_j^{(1)}\rangle... |\beta_j^{(s)} \rangle$ from the state $|\psi_{A^{(0)}}\rangle$. Then followed by controlled rotation, uncomputing and measurement, we could obtain the desired state $|\psi_{A^{(s)}}\rangle$. However, it requires much more space resource than DYXL algorithm. To reduce the space complexity, we transform the register $C$ to $|0\rangle$ and only keep $|\beta_j^{(i)}\rangle$ after iteration $i$, i.e., obtain state $|\psi_{i}\rangle =\frac{1}{\sqrt{k}}\sum_{j=0}^{k-1}|\textbf{u}_j\rangle |\textbf{j}\rangle|\sigma_j^2 \rangle^B |\beta_j^{(i)} \rangle^D |0\rangle^C$.

(\romannumeral2) Perform quantum arithmetic operation on register $B$, $C$, and $D$, to get $|\psi_{i}\rangle$.

 Since $c^{(i)} \beta_j^{(i)} = \frac{(\sigma_j \beta_j^{(i-1)})^2+ \lambda_2}{\sigma_j^2 \beta_j^{(i-1)}}$, we have
\begin{eqnarray}
    \sigma_j^2 (\beta_j^{(i-1)})^2 - c^{(i)}\sigma_j^2 \beta_j^{(i)} \beta_j^{(i-1)}+ \lambda_2=0,
\end{eqnarray}
which is a one-variable quadratic equation about the variable $\beta_j^{(i-1)}$. The solutions of the equation is
\begin{eqnarray}
    \beta_{j\pm}^{(i-1)} = \frac {c^{(i)}\sigma_j^2 \beta_j^{(i)} \pm \sqrt{(c^{(i)}\sigma_j^2 \beta_j^{(i)})^2 - 4\sigma_j^2 \lambda_2}}{2\sigma_j^2}.
\end{eqnarray}

Two cases are considered here. In case 1, when $\lambda_2 \ge 1$, then for $x \le 1$, the function $f(x)=\frac{(\sigma_j x)^2+\lambda_2}{\sigma_j^2 x}$ is a monotonic decreasing function, which means that $\beta_{j}^{(i-1)} = \beta_{j-}^{(i-1)}$. In case 2, when $\lambda_2 < 1$, we add a qubit (register $F$) to store the magnitude relationship between $\beta_j^{(i-1)}$ and $\sqrt{\frac{\lambda}{\sigma^2}}$ on the state in equation (\ref{eq:22}), i.e.,
    \begin{eqnarray*}
    \begin{split}
        |\phi_4^{i}\rangle =\frac{1}{\sqrt{k}}\sum_{j=0}^{k-1}|\textbf{u}_j\rangle |\textbf{j}\rangle|\sigma_j^2 \rangle^B |\beta_j^{(i-1)}\rangle^C |\beta_j^{(i)} \rangle^D |\gamma_j^{(i)}\rangle^F,
    \end{split}
    \end{eqnarray*}
where
    \begin{eqnarray}
    \begin{split}
        |\gamma_j^{(i)}\rangle =
            \begin{cases}
                |1\rangle, & if \mbox{ } \beta_j^{(i-1)} \ge \sqrt{\frac{\lambda}{\sigma^2}}, \\
                |0\rangle, & if \mbox{ } \beta_j^{(i-1)} < \sqrt{\frac{\lambda}{\sigma^2}}.
            \end{cases}
    \end{split}
    \end{eqnarray}
Then, according to the information in register $F$, we could get
    \begin{eqnarray}
    \begin{split}
        \beta_{j}^{(i-1)} =
            \begin{cases}
                  \beta_{j+}^{(i-1)}, & if \mbox{ } |\gamma_j^{(i)}\rangle = |1\rangle,\\
                  \beta_{j-}^{(i-1)}, & if \mbox{ } |\gamma_j^{(i)}\rangle = |0\rangle.
            \end{cases}
    \end{split}
    \end{eqnarray}
Thus we could transform the state of register $C$ to $|0\rangle$ by a simple quantum arithmetic operation on register $C$ and $D$.
We should keep in mind that we need an ancillary qubit in each iteration for case 2.

5.  \emph{Iteration} For $i=1$ to $s-2$, repeat step 2 to 4. Thus we obtain state
    \begin{eqnarray}
    \begin{split}
        |\psi_{s-1}\rangle =\frac{1}{\sqrt{k}}\sum_{j=0}^{k-1}|\textbf{u}_j\rangle |\textbf{j}\rangle|\sigma_j^2 \rangle |\beta_j^{(s-1)} \rangle |0\rangle.
    \end{split}
    \end{eqnarray}

 6. \emph{Controlled rotation} Add an ancillary qubit (register $E$) and perform an appropriate controlled rotation on the state $|\psi_{s-1}\rangle|0\rangle^E$, transforms the system to
    \begin{eqnarray}
    \label{eq:eq23}
     \begin{split}
        |\phi_2^{(s)}\rangle=&\frac{1}{\sqrt{k}}\sum_{j=0}^{k-1}|\textbf{u}_j\rangle |\textbf{j}\rangle|\sigma_j^2 \rangle |\beta_j^{(s-1)}\rangle |c^{(s)}\beta_j^{(s)} \rangle \\
        &(\sqrt{1-(\frac{c^{(s)}\beta_j^{(s)}}{c})^2}|0\rangle+\frac{c^{(s)}\beta_j^{(s)}}{c}|1\rangle)^E.
     \end{split}
     \end{eqnarray}

 7. \emph{Uncomputing and measurement} Uncompute register $B$, $C$, and $D$, and measure the register $E$ to seeing 1, thus obtain
    \begin{eqnarray}
    |\psi_{A^{(s)}}\rangle =\sum_{j=0}^{k-1}\beta_j^{(s)}|\textbf{u}_j\rangle |\textbf{j}\rangle.
    \end{eqnarray}

\subsection{The complexity of the improved quantum AOP algorithm}
We have described an improved quantum AOP algorithm above. In this subsection, we will analyze the time complexity and space complexity of the algorithm.

The time complexity and space complexity of step 1 are $O(\log_2 (\epsilon^{-1}) \log_2(nk))$ and $O(\log_2(nk/\epsilon))$, the same as the DYXL algorithm.

In step 2, similar to the complexity of the step 3 of the DYXL algorithm, the phase estimation stage is of time complexity $O(\frac{1}{\epsilon_1} \mathrm{polylog} (\frac{nm}{\epsilon_1}))$ and space complexity $O(\log_2(1/\epsilon_1))$ with error $\epsilon_1$. The stage of appending registers $|\frac{1}{\sqrt{k}} \rangle^C |0 \rangle^D $ is of time complexity $O(\log_2(1/\epsilon_1))$, where the number of qubits in register $B$, $C$ and $D$ is $O(\log_2(1/\epsilon_1))$. Thus the time complexity of this step is $O(\frac{1}{\epsilon_1} \mathrm{polylog} (\frac{nm}{\epsilon_1}))$.

In step 3, since the time complexity of preparing the state $|\psi_{i-1}\rangle$ is much greater than the complexity of stage (\romannumeral1) and stage (\romannumeral2) (which is $O(\mathrm{polylog}(1/\epsilon_1))$ and $O(\log_2(1/\epsilon_1))$ respectively), we will neglect the complexity of these two stages. In stage (\romannumeral3), define
\begin{eqnarray*}
\begin{split}
U_{i-1}&: U_{i-1}|0\rangle= |\phi_2^{(i)}\rangle =\sin(\theta)|a\rangle|0\rangle + \cos(\theta) |b\rangle |1\rangle, \\
S_0&: S_0=I-2|0\rangle^{ABCDE}\langle0|^{ABCDE},\\
S_\chi&: S_\chi=I-2|1\rangle^{E}\langle1|^{E}.
\end{split}
\end{eqnarray*}
According to quantum amplitude estimation algorithm \cite{NC, BHMT}, the unitary operator $Q = -U_{i-1} S_0 U_{i-1}^\dag S_\chi$ act as a rotation on the two dimensional space $\mathrm{Span}\{|a\rangle|0\rangle, |b\rangle |1\rangle\}$, with eigenvalues $e^{\pm2i\theta}$ and corresponding eigenvectors $\frac{|a\rangle|0\rangle\mp|b\rangle |1\rangle}{2}$.
The quantum amplitude estimation algorithm will generate $\theta$ within error $\epsilon_2$, which means that $O(\log_2(1/\epsilon_2))$ qubits are required to store $\theta$. The corresponding time complexity is $O(\frac{1}{\epsilon_2}(2+\frac{1}{2\eta})G_{i-1})$, where $1-\eta$ is the probability to success (we could simply choose $\eta=O(1))$.
Finally, we obtain the classical information of $c^{(i)}$ within relative error $O(\kappa^2\epsilon_2)$ (see appendix \ref{app:D}), here we use relative error to ensure that the estimation of $c^{(i)}$ won't be influence by the scale of $c^{(i)}$.

In step 4, for the stage (\romannumeral1), similar to the analysis of the DYXL algorithm (see appendix \ref{app:complexity}), we want to bound the relative error of $\beta_j^{(i)}$ (denote as $\tilde{\epsilon}_{\beta}$) by $O(\epsilon)$. According to Appendix \ref{app:D}, $\tilde{\epsilon}_{\beta}=O(\kappa^2 \sqrt{k}\epsilon_1+\tilde{\epsilon}_{c}) = O(\kappa^2\sqrt{k} \epsilon_1+\kappa^2\epsilon_2)$, thus we can choose $\epsilon_1=O(\frac{\epsilon} {\kappa^2\sqrt{k}})$ and $\epsilon_2=O(\frac{\epsilon}{\kappa^2})$ to ensure $\tilde{\epsilon}_{\beta}=O(\epsilon)$. Since the classical information of $c^{(i)}$ is given by step 3, this stage is of time complexity $O(G_{i-1} + \mathrm{polylog}(\frac{\kappa k}{\epsilon}))$. As for the stage (\romannumeral2), the time complexity of the two cases are $O(\mathrm{polylog}(\frac{\kappa k}{\epsilon}))$. In the worst case (case 2 of stage (\romannumeral2)), we need an ancillary qubit (register $F$) in each iteration.

In step 6, the time complexities of the two operations are $O(\mathrm{polylog}(\epsilon_1))=O(\mathrm{polylog}(\frac{\kappa k}{\epsilon}))$ and $O(\log_2(1/\epsilon))$, which is similar with the stage (\romannumeral1) and stage (\romannumeral2) of step 3. Note that no extra qubit is needed here, since register $E$ is reused.

In step 7, the time complexity of the uncomputing stage is just the same as the time complexity to prepare state $|\phi_2^{(s)}\rangle$ (equation (\ref{eq:eq23})) from state $|\psi_{A^{(0)}}\rangle$.
For $c^{(s)}=\Omega(k)$ and $c=O(\sqrt{k}\kappa^2)$, the probability of seeing 1 is
    \begin{eqnarray}
    \begin{split}
        p(1) = \frac{1}{k}\sum_{j=0}^{k-1} (\frac{c^{(s)}}{c}\beta_j^{(s)})^2 = \Omega(\frac{1}{\kappa^4}).
    \end{split}
    \end{eqnarray}
We can use the quantum amplitude amplification \cite{NC, BHMT} to reduce the repetition to $O(\kappa^2)$.

Now we put all together. From the analysis of step 3, we know that the time complexity to estimate $c^{(i)}, i=1,...,s$ is $O_{c^{(i)}}=O(\frac{1}{\epsilon_2} G_{i-1})$. Also, from step 4, given $c^{(i)}$, the time complexity to prepare state $|\psi_{i}\rangle$ is
    \begin{eqnarray}
    \begin{split}
        G_{i} = G_{i-1}+2\mathrm{polylog}(\frac{\kappa k}{\epsilon}).
    \end{split}
    \end{eqnarray}
Thus given $c^{(i)}$ for $i=1,2,...,s-1$,
    \begin{eqnarray}
    \begin{split}
        G_{s-1} = G_{0}+2(s-1)\mathrm{polylog}(\frac{\kappa k}{\epsilon}),
    \end{split}
    \end{eqnarray}
where $G_0 = \frac{\kappa^2 \sqrt{k}}{\epsilon} \mathrm{polylog}(\frac{nm}{\epsilon})$.

The time complexity to obtain state $|\psi_{A^{(s)}}\rangle$ is as follows:
    \begin{eqnarray}
    \begin{split}
        O_{|\psi_{A^{(s)}}\rangle}&=O(\kappa^2 2O_{|\phi_2^{(s)}\rangle}) \\
        & =O(\kappa^2 (\mathrm{polylog}(\frac{\kappa k}{\epsilon})+G_{s-1} + \sum_{i=1}^{s-1} c^{(i)}))\\
        & = O(\frac{s \kappa^6 \sqrt{k}}{\epsilon}\mathrm{polylog}(\frac{nm}{\epsilon}) + \frac{s^2 \kappa^4}{\epsilon}\mathrm{polylog}(\frac{\kappa k}{\epsilon})).
    \end{split}
    \end{eqnarray}

As for the space complexity, $O(\log_2(nk/\epsilon))$ qubits are required to obtain the classical information of $c^{(i)}$ in step 3. Given state $|\psi_{i-1}\rangle$, an ancillary qubit is required to prepare state $|\psi_i\rangle$ in step 4 in iteration $i$ when $\lambda_2 \le 1$. if $\lambda_2 \ge 1$, no ancillary qubit is needed. Since this ancillary qubit cannot be reused by the other iterations, $O(s)$ ancillary qubits are required for all iterations. In step 5 to step 7, no ancillary qubit is needed. Putting all together, the space complexity of the algorithm is $O(\log_2(nk/\epsilon)+s)$.

Since our algorithm and the DYXL algorithm are iterative algorithms, we compare the number of iterations of the two quantum algorithms with the same loss. In each iteration of the DYXL algorithm, given the state $|\psi_{A^{(i-1)}}\rangle$, we can obtain the state $|\psi_{A^{(i)}}\rangle$ within a relative error $O(\epsilon)$ in $\beta_j^{(i)}$ (Appendix \ref{app:complexity}). And in each iteration of our algorithm, given $|\psi_{i-1}\rangle$, we obtain the state $|\psi_{i}\rangle$ within a relative error $O(\epsilon)$ in $\beta_j^{(i)}$, the same as the DYXL algorithm. Thus the two quantum algorithms will converge to the same loss with the same number of iterations.

\section{discussion}

The exponential speedup claimed by Duan \emph{et al}. \cite{DYXL} is based on the assumptions that $\kappa$, $k$ and $1/\epsilon$ are $O(\mathrm{polylog}(nm))$. We follow these assumptions to compare the two quantum algorithms. The advantage of our algorithm is that the time complexity is quadratically dependent on $s$, while the DYXL algorithm has an exponential dependence on $s$. If $s$ is a constant, our algorithm has a polynomial speedup over the DYXL algorithm. If $s$ grows linearly with $\log_2(mn)$, our algorithm has exponential speedups on $n$ and $m$ compared with the DYXL algorithm. Since the speedup of our algorithm is strongly dependent on $s$, we now analyze the value of $s$ below.

Note that the steps of the two quantum algorithms are exactly the same as those of the reformulated AOP algorithm in Section \ref{sub:2b}. Thus by controlling the precisions of each iteration to the same level, the number of iterations of the quantum algorithms is the same as the reformulated AOP algorithm. We estimate the number of iterations of the reformulated AOP algorithm in Appendix \ref{app:iteration} by numerical experiments on randomly generated datasets, since it is difficult to determine the value of $s$ through theoretical analysis. The experimental results show that $s=\Omega(k+\kappa+\log_2(1/\epsilon))$ may hold. If it holds, $s$ grows linearly with $\log_2(mn)$, which means that the exponential speedups on $n$ and $m$ of our algorithm hold. Note that the experimental results are based on randomly generated original datasets, it can't rule out the possibility that the parameter $s$ may be less in the practical datasets.

\section{Conclusion}
In this paper, we reanalyzed the DYXL algorithm in detail and corrected the complexity calculation. It was shown that the DYXL algorithm has an exponential dependence on the number of iterations $s$, thus the quantum algorithm may lose its advantage as $s$ increases. To get a further speedup, we presented an improved quantum AOP algorithm with a time complexity quadratic on $s$. Our algorithm achieves at least a polynomial speedup over the DYXL algorithm. When both $\kappa$, $k$ and $1/\epsilon$ are $O(\mathrm{polylog}(nm))$, our algorithm achieves exponential speedups compared with the classical algorithm on $n$ and $m$. As for the space complexity, our algorithm is slightly worse than DYXL algorithm.

The speedups of our algorithm mainly come from the idea of putting the information to be updated into computational basis, which saves the consumption of the current candidates. We hope this idea could inspire more iterative algorithms to get a quantum speedup. We will explore the possibility in the future.

\section*{Acknowledgements}
We would like to thank the anonymous referees for their helpful comments. S-J Pan thanks Chun-Guang Li, Chun-Tao Ding and Sheng-Jie Li for fruitful discussions. This work was supported by
the National Natural Science Foundation of China (Grants No. 61672110, No. 61671082, No. 61976024, No. 61972048, and No. 61801126), the Fundamental Research Funds for the Central Universities (Grant No. 2019XDA01), the Open project of CAS Key Laboratory of Quantum Information, University of Science and Technology of China (Grant No. KQI201902) and the Beijing Excellent Talents Training Funding Project (Grant No. 201800002685XG356).

\appendix
\section{The complexity of each step of the DYXL algorithm}
\label{app:complexity}

In this appendix, we analyze the time complexity of each step of the DYXL algorithm in detail. Different from the original paper \cite{DYXL}, on the one hand, we use the best-known results on Hamiltonian simulation to get a tight bound of the time complexity, on the other hand, we estimate the parameters which have not been estimated in the original paper or need to be correct.

In step 1, $ |\psi_{A^{(0)}}\rangle$ can be written in computational basis as $\sum_{x_1 x_2... x_l \in \{0,1\}^{l}} \alpha_{x_1 x_2... x_l}|x_1 x_2... x_l\rangle$, where $l=k\log_2 n$. By the assumption that the elements of $A^{(0)}$ and $\omega^{(i)}$ are given and stored in QRAM, where $\omega^{(i)}$ is defined as $\cos^2(2\pi\omega_i)=(\frac{\alpha_{x_1 x_2... x_{i-1} 0}}{{x_1 x_2... x_{i-1}}})^2+O(\mathrm{poly}(\epsilon))$,
the time complexity and space complexity for preparing $|\psi_{A^{(0)}}\rangle$ are $O(\log_2 (\epsilon^{-1}) \log_2(nk))$ and $O(\log_2(nk/\epsilon))$ \cite{DYXL}.

In step 3, Assume the condition number of $\widetilde{X}$ and $A^{(i-1)}$ is $\kappa$ and $\kappa^{(i-1)}$, then $\kappa^{(i-1)} = O(\kappa^2)$ (the proof is given in appendix \ref{app:A}. We should mention that in the original paper \cite{DYXL}, the author gave a wrong estimate of $\kappa^{(i-1)}$, which influenced the total complexity). Note that
\begin{eqnarray}
\label{eq:equ14}
\begin{split}
\mathrm{tr}_2(|\psi_{A^{(i-1)}}\rangle\langle \psi_{A^{(i-1)}|)} &= \sum_{j=1}^k (\beta_j^{(i-1)})^2 |\mathbf{u}_j \rangle \langle \mathbf{u}_j| \\
&=A^{(i-1)}{A^{(i-1)}}^\dagger.
\end{split}
\end{eqnarray}
According to \emph{Corollary} \ref{Hampure} (\emph{Corollary} 17 of \cite{LC}), the time complexity to simulate $e^{iA^{(i-1)}{{A^{(i-1)}}^\dagger t_0}}$ within error $\epsilon_0$ is $O\left((t_0+ \mathrm{log_2} (1/\epsilon_0)) G^{(i-1)}\right)$, where $G^{(i-1)}$ is the time complexity to prepare state $|\psi_{A^{(i-1)}}\rangle$.
For the simulation of $\widetilde{X} \widetilde{X}^ \dagger$, assume there is a quantum circuit to prepare state $|\psi_{\widetilde{X}}\rangle = \frac{1}{\|\widetilde{X}\|_F} \sum_{i=0}^{n-1} \sum_{j=0}^{m-1} \widetilde{X}_{ij} |i\rangle |j\rangle = \frac{1}{\|\widetilde{X}\|_F}\sum_{j=0}^{r-1} \sigma_j |\mathbf{u}_j\rangle |\mathbf{v}_j\rangle$ in time $O(\mathrm{polylog}(nm))$. Notice that $\widetilde{X} \widetilde{X}^ \dagger = \|\widetilde{X}\|_F^2\mathrm{tr}_2(|\psi_{\widetilde{X}}\rangle\langle \psi_{\widetilde{X}}|)$, where $\|\widetilde{X}\|_F = \sqrt{\sum_{j=0}^{r-1} \sigma_j^2} \le \sqrt {r} = O(\mathrm{polylog}(nm))$. Thus the time complexity to simulate $e^{i\widetilde{X} \widetilde{X}^ \dagger t_0}$ within $\epsilon_0$ is $O \left(\mathrm{polylog} (nm) (\|\widetilde{X}\|_F^2 t_0 + \mathrm{log_2} (1 / \epsilon_0)) \right)$.
\begin{corollary}(\cite{LC})
\label{Hampure}
Given access to the oracle $\hat{G}$ specifying a Hamiltonian $\hat{H}=\hat{\rho}$ that is a density matrix $\hat{\rho}$, where
\begin{eqnarray}
\label{eq:equ13}
\begin{split}
&\hat{G}|0\rangle_a=|G\rangle_a=\sum_{j}\sqrt{\alpha_j}|j\rangle_{a_1}|\chi_j\rangle_{a_2},\\
&\hat\rho = \mathrm{tr}{|G\rangle\langle G|_{a_1}=\sum_{j}\alpha_j|\chi_j\rangle\langle\chi_j|},
\end{split}
\end{eqnarray}
time evolution by $\hat{H}$ can be simulated for time $t$ and error $\epsilon$ with $\mathcal{O}(t+\log_2{(1/\epsilon)})$ queries.
\end{corollary}
According to \cite{CEMM,LP,NC}, taking $O(1/\epsilon_1)$ times of controlled-$e^{iA^{(i-1)}{{A^{(i-1)}}^{\dagger} t_0}}$ to perform phase estimation ensures that the eigenvalues $\beta_j^{(i-1)}$ being estimated within error $O(\epsilon_1)$, so as the phase estimation on $e^{i\widetilde{X} \widetilde{X}^ \dagger t_0}$. Let $\epsilon_0= \epsilon_1^2$ and $t_0=O(1)$, the time complexity of the two phase estimations are $O\left((G^{(i-1)}/\epsilon_1) \log_2(1/\epsilon_1)\right)$ and $O\left((1/\epsilon_1) \mathrm {polylog}(nm/\epsilon_1)\right)$ respectively, while the space complexity of these two phase estimations are $O(\log_2(1/\epsilon_1))$.

The implementation of controlled rotation of step 4 can be divide into two stages \cite{DYXL}. The first stage is a quantum circuit to compute $y_j=\rho\frac{(\sigma_j \beta_j^{(i-1)})^2+\lambda_2}{(\sigma_j {\beta_j^{(i-1)}})^2}$ and store in an auxiliary register $L$ with $O(\mathrm{log_2}(1/\epsilon))$ qubits, where $\rho=O(\frac{1}{2\lambda_2 k \kappa^4})$ (the proof is given in appendix \ref{app:B}, we should mention that in the original paper \cite{DYXL}, the authors did not analyze the parameter $\rho$ which has a strong correlation with the total complexity).
Since $\beta_j^{(i-1)}$ and $\sigma_j$  is estimate with error $O(\epsilon_1)$, the relative error of estimating $y_j$ is
\begin{eqnarray}
\begin{split}
\tilde{\epsilon}_{y}
&= O\left(\frac{\lambda_2 (\beta_j^{(i-1)})^2 \sigma_j+\lambda_2 \beta_j^{(i-1)} \sigma_j^2}{(\beta_j^{(i-1)})^4+\sigma_j^4+\lambda_2 (\beta_j^{(i-1)})^2 \sigma_j^2}\epsilon_1 \right) \\
&= O\left( \frac{\lambda_2(\beta_j^{(i-1)}+\sigma_j)\beta_j^{(i-1)}\sigma_j}{(\lambda_2+2) (\beta_j^{(i-1)})^2 \sigma_j^2}\epsilon_1  \right) \\
&= O(\kappa^2\sqrt{k}\epsilon_1),
\end{split}
\end{eqnarray}
in the first equation we neglect the terms with the power of $\epsilon_1$ greater than $1$ and in the last equation we use the conclusion of equation (\ref{eq:B1}).
To ensure that the final error of this iteration is within $O(\epsilon)$, we could take $\tilde{\epsilon}_y=O(\epsilon)$, which means $\epsilon_1=O(\frac{\epsilon}{\kappa^2\sqrt{k}})$.
Following the result of \cite{CPPTK}, the time complexity of this stage is $O(\mathrm{polylog} (1/\epsilon)))$. The second stage is to perform controlled rotation $CR$ on the state, where $CR|y_j\rangle^L |0\rangle^D = |y_j\rangle^L (\sqrt{1-y_j^2}|0\rangle + y_j|1\rangle)^D$. The time complexity of this stage is $O(\log_2(1/\epsilon))$ \cite{HHL,DYLL,WBL,SSP, YGW, YGLHRW}.

In step 5, the probability to seeing 1 in register D is $p(1) =\rho^2 \sum_j \left( \frac{(\sigma_j \beta_j^{(i-1)})^2+ \lambda_2 }{\sigma_j^2 \beta_j^{(i-1)}} \right)^2 = O(\frac{1}{\kappa^4})$
as shown in Appendix \ref{app:B} (We should mention that in the original paper \cite{DYXL}, the authors did not analyze this probability which is directly related to the total complexity). Using amplitude amplification \cite{BHMT}, we find
that $O(\kappa^2)$ repetitions are sufficient.

\section{Estimate the parameter $\kappa^{(i)}$}
\label{app:A}
In this appendix, we analyze the condition number $\kappa^{(i)}$ of $A^{(i)}$.
Firstly, we give the following theorem:

\begin{theorem}\label{the:2}
If $\beta_j^{(i-1)}>\beta_{j'}^{(i-1)}$, then $\beta_j^{(i)}>\beta_{j'}^{(i)}$.
\end{theorem}
\begin{proof}\renewcommand{\qedsymbol}{}
(1) Note that $\beta_0^{(0)}=\beta_1^{(0)}=...= \beta_{k-1}^{(0)}$ and $1 = \sigma_0 \ge \sigma_1\ge ... \ge \sigma_{k-1} = \frac{1}{\kappa} $. For $j>j'$, the following inequalities hold:
\begin{eqnarray}
\begin{split}
    \beta_j^{(0)} &\ge \beta_{j'}^{(0)}, \\
    \sigma_j^2 \beta_j^{(0)} &\le \sigma_{j'}^2 \beta_{j'}^{(0)}.\\
\end{split}
\end{eqnarray}

(2) Assuming that for $i \ge 1$, the following inequalities hold:
\begin{eqnarray}
\begin{split}
    \beta_j^{(i-1)} &\ge \beta_{j'}^{(i-1)}, \\
    \sigma_j^2 \beta_j^{(i-1)} &\le \sigma_{j'}^2 \beta_{j'}^{(i-1)}.\\
\end{split}
\end{eqnarray}
 Then
 \begin{eqnarray}
\begin{split}
    \sigma_j^2 \beta_j^{(i)} &= \sigma_j^2 \frac{(\sigma_j \beta_j^{(i-1)})^2 + \lambda_2 }{c^{(i-1)}\sigma_j^2 \beta_j^{(i-1)} } \\
    &= \frac{1}{c^{(i-1)}} \sigma_j^2 \beta_j^{(i-1)} + \frac{\lambda_2}{c^{(i-1)} \beta_j^{(i-1)}} \\
    &\le  \frac{1}{c^{(i-1)}} \sigma_{j'}^2 \beta_{j'}^{(i-1)} + \frac{\lambda_2}{c^{(i-1)} \beta_{j'}^{(i-1)}} \\
    &=  \sigma_{j'}^2 \beta_{j'}^{(i)}.
\end{split}
\end{eqnarray}
 Also,
\begin{eqnarray}
\begin{split}
    \beta_j^{(i)} &= \frac{(\sigma_j \beta_j^{(i-1)})^2+ \lambda_2}{c^{(i-1)}\sigma_j^2 \beta_j^{(i-1)}} \\
    &=\frac{1}{c^{(i-1)}}\left(\beta_j^{(i-1)} + \frac{\lambda_2}{\sigma_j^2 \beta_j^{(i-1)}} \right) \\
    &\ge \frac{1}{c^{(i-1)}}\left(\beta_{j'}^{(i-1)} + \frac{\lambda_2}{\sigma_{j'}^2 \beta_{j'}^{(i-1)}} \right)\\
    &=\beta_{j'}^{(i)}.
\end{split}
\end{eqnarray}

Thus the theorem holds.
\end{proof}

According to \emph{Theorem} \ref{the:2}, we have:
\begin{eqnarray}
\begin{split}
\max_j \beta_j^{(i)}&=\frac{\max_j (\beta_j^{(i-1)})^2 + \lambda_2 \kappa^2}{c^{(i-1)}\max_j \beta_j^{(i-1)}}, \\
\min_j \beta_j^{(i)}&=\frac{\min_j (\beta_j^{(i-1)})^2 + \lambda_2}{c^{(i-1)}\min_j \beta_j^{(i-1)}}.
\end{split}
\end{eqnarray}
Let $\max_j(\beta_j^{(i)})^2 = a^{(i)} \le 1$, thus,
\begin{eqnarray}
\label{eq:eqA1}
\begin{split}
\kappa^{(i)} &= \frac{\max_j \beta_j^{(i)}}{\min_j \beta_j^{(i)}} \\
& = \frac{\max_j(\beta_j^{(i-1)})^2 + \lambda_2 \kappa^2}{c^{(i-1)} \max_j \beta_j^{(i-1)}} \frac{c^{(i-1)}\min_j \beta_j^{(i-1)}} {\min_j (\beta_j^{(i-1)})^2 + \lambda_2} \\
& = \frac{(a^{(i-1)} + \lambda_2 \kappa^2)\kappa^{(i-1)}}{a^{(i-1)} + \lambda_2(\kappa^{(i-1)})^2}.
\end{split}
\end{eqnarray}

Note that $\beta_j^{(0)}$ takes the same value for $j=0,1,...k-1$, thus $\kappa^{(0)} = 1$, $\kappa^{(1)} = \frac{a^{(0) + \lambda_2 \kappa^2}}{a^{(0)} + \lambda_2}$. According to equation (\ref{eq:eqA1}), the following equation holds:
\begin{eqnarray}
\label{eq:eqA2}
\begin{split}
\kappa^{(i-1)}\kappa^{(i)} = \frac{(a^{(i-1)} + \lambda_2 \kappa^2)(\kappa^{(i-1)})^2}{a^{(i-1)} + \lambda_2(\kappa^{(i-1)})^2}.
\end{split}
\end{eqnarray}
According to equation (\ref{eq:eqA1}) and  (\ref{eq:eqA2}), we have:
\begin{eqnarray}
\begin{split}
 &\kappa^{(i)}\kappa^{(i-1)} \ge \kappa^2, \quad &if \quad \kappa^{(i-1)} \ge \kappa,\\
 &\kappa^{(i)}\kappa^{(i-1)} \le \kappa^2, \quad &if \quad \kappa^{(i-1)} \le \kappa,
\end{split}
\end{eqnarray}
Since $\kappa=O(\mathrm{polylog}(NM))$, $\lambda_2=O(1)$, we assume that $\lambda_2 \kappa > 1$, thus $\kappa^{(0)}=1 \le \kappa$, $\kappa^{(1)}=\frac{a^{(0) + \lambda_2 \kappa^2}}{a^{(0)} + \lambda_2} \ge \kappa$. Then
\begin{eqnarray}
\begin{split}
 &\kappa^{(2l)}\kappa^{(2l+1)} \le \kappa^2 \\
 &\kappa^{(2l+1)}\kappa^{(2l+2)} \ge \kappa^2.
\end{split}
\end{eqnarray}
Further,
\begin{eqnarray}
\begin{split}
 1&= \kappa^{(0)} \le \kappa^{(2)} \le ... \le \kappa^{(2l)} \le \kappa \\
 &\le \kappa^{(2l+1)} \le ... \le \kappa^{(3)} \le \kappa^{(1)} \\
 &=\frac{a^{(0) + \lambda_2 \kappa^2}}{a^{(0)} + \lambda_2} = O(\kappa^2).
 \end{split}
\end{eqnarray}
The sequence $\kappa^{(2l)}$ ($l=0,1,2,...\infty$) is monotonically increasing with upper bound $\kappa$, and the sequence $\kappa^{(2l+1)}$ ($l=0,1,2,...\infty$) is monotonically decreasing with lower bound $\kappa$, thus the sequence $\kappa^{(i)}$ ($i=0,1,2,...\infty$) converges on $\kappa$, i.e. $\lim_{i\rightarrow \infty} \kappa^{(i)}=\kappa$.

In conclusion, We get $\kappa^{(i)}=O(\kappa^2)$ for $i \ge 0$.

\section{Estimate the parameter $\rho$ and the $p(1)$ of step 5 in the DYXL algorithm}
\label{app:B}
In this appendix, we analyze the value of parameter $\rho$ which first appeared in step 4 of the DYXL algorithm.

It is obvious that $\sum_{j=1}^k (\beta_j^{i-1})^2=1$, for $\beta_j^{(i-1)} (j=0,1,...,k-1)$ is the amplitude of the quantum state $|\psi_A^{(i-1)}\rangle$. Thus,
\begin{eqnarray}
\begin{split}\label{eq:B1}
\frac{1}{\sqrt{k}} &\le \max_j \beta_j^{(i-1)} \le 1, \\
\frac{1}{\kappa^{i-1}\sqrt{k}} &\le \min_j \beta_j^{(i-1)} \le \frac{1}{\sqrt{k}}.
\end{split}
\end{eqnarray}
Since $\rho \frac{(\sigma_j \beta_j^{(i-1)})^2+\lambda_2}{ (\sigma_j {\beta_j^{(i-1)}})^2} \le 1$, $\rho \le \min_j \frac{ (\sigma_j {\beta_j^{(i-1)}})^2}{(\sigma_j \beta_j^{(i-1)})^2+\lambda_2}$. Note that
\begin{eqnarray}
\begin{split}
&\min_j \frac{ (\sigma_j {\beta_j^{(i-1)}})^2}{(\sigma_j \beta_j^{(i-1)})^2+\lambda_2}
= \min_j \frac{\beta_j^{(i-1)}}{\beta_j^{(i-1)}+\frac{\lambda_2}{\sigma_j^2 \beta_j^{(i-1)}}} \\
&\ge \frac{\min_j \beta_j^{(i-1)}}{\max_j \left(\beta_j^{(i-1)}+ \frac{\lambda_2} {\sigma_j^2 \beta_j^{(i-1)}} \right)} \\
& = \frac{\min_j \beta_j^{(i-1)} \max_j \beta_j^{(i-1)}}{\max_j \left(\beta_j^{(i-1)}\right)^2 + \lambda_2 \kappa^2}
 \ge \frac{1/(\kappa^{(i-1)}k)}{1+\lambda_2\kappa^2} \\
& = \frac{1}{\kappa^{(i-1)}k+\lambda_2 k \kappa^2 \kappa^{(i-1)}}
> \frac{1}{2 \lambda_2 k \kappa^2 \kappa^{(i-1)}},
\end{split}
\end{eqnarray}
for the $\kappa^{(i)}=O(\kappa^2)$, we could choose the parameter $\rho=O(\frac{1}{2\lambda_2 k \kappa^4})$. Then the probability to seeing 1 in register D in step 5 is
\begin{eqnarray}
\begin{split}
p(1) &=\rho^2 \sum_j \left( \frac{(\sigma_j \beta_j^{(i-1)})^2+ \lambda_2 }{\sigma_j^2 \beta_j^{(i-1)}} \right)^2 \\
&\le \rho^2 k \left( \frac{(\frac{1}{\kappa} \max_j \beta_j^{(i-1)})^2+ \lambda_2 }{\frac{1}{\kappa^2} \max_j \beta_j^{(i-1)}} \right)^2 \\
&\le \rho^2 k \left(\frac {1+\lambda_2\kappa^2}{\max_j \beta_j^{(i-1)}} \right)^2 \\
&\le \rho^2 k^2 (1+\lambda_2 \kappa^2)^2 \\
&\le (\frac{1}{2\lambda_2 k \kappa^4})^2 k^2 (1+\lambda_2 \kappa^2)^2 \\
&= O(\frac{1}{\kappa^4}).
\end{split}
\end{eqnarray}

\begin{figure*}
\centering
\begin{minipage}[c]{0.4\textwidth}
\centering
\includegraphics[height=5.5cm,width=7.5cm]{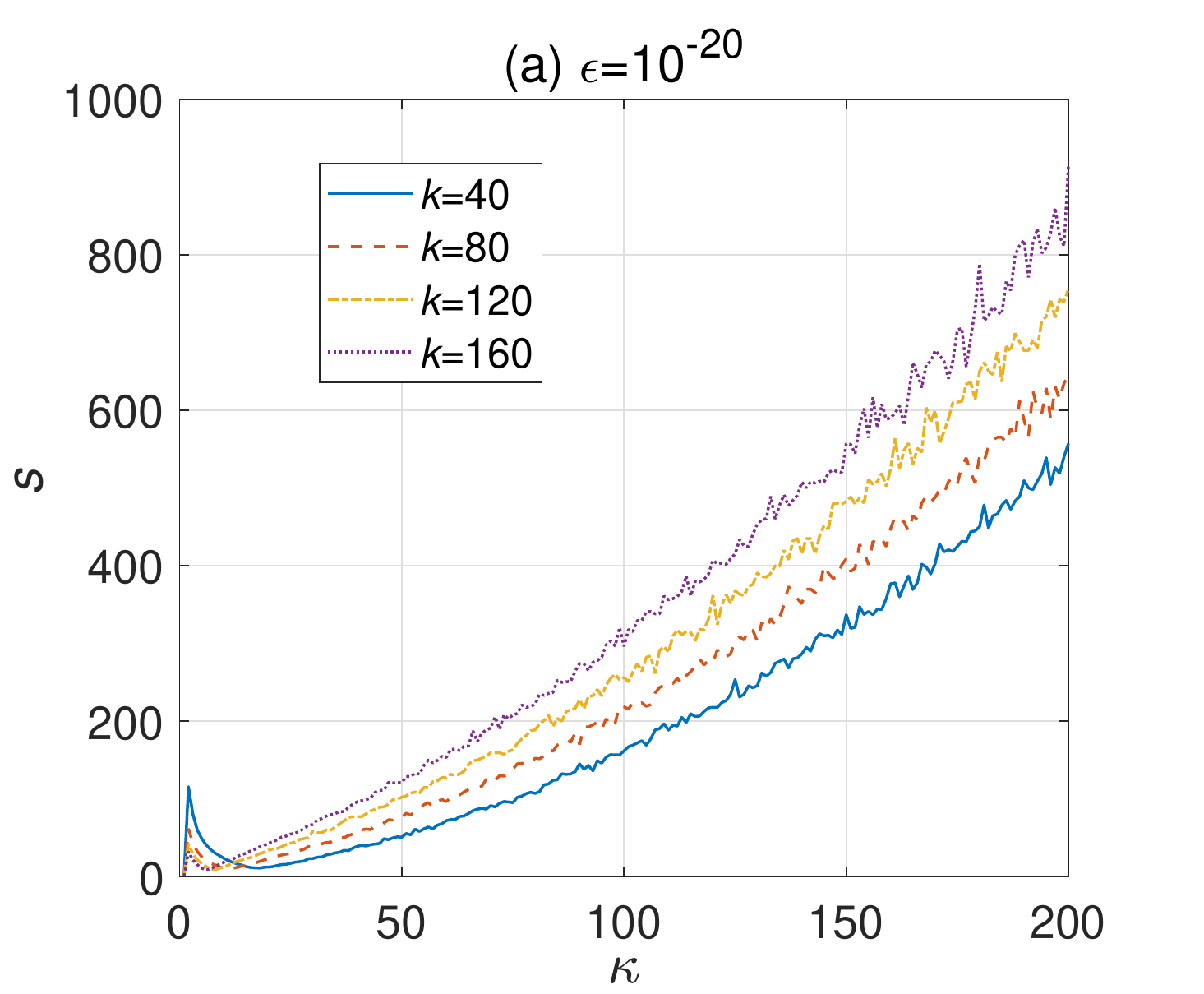}
\end{minipage}%
\begin{minipage}[c]{0.4\textwidth}
\centering
\includegraphics[height=5.5cm,width=7.5cm]{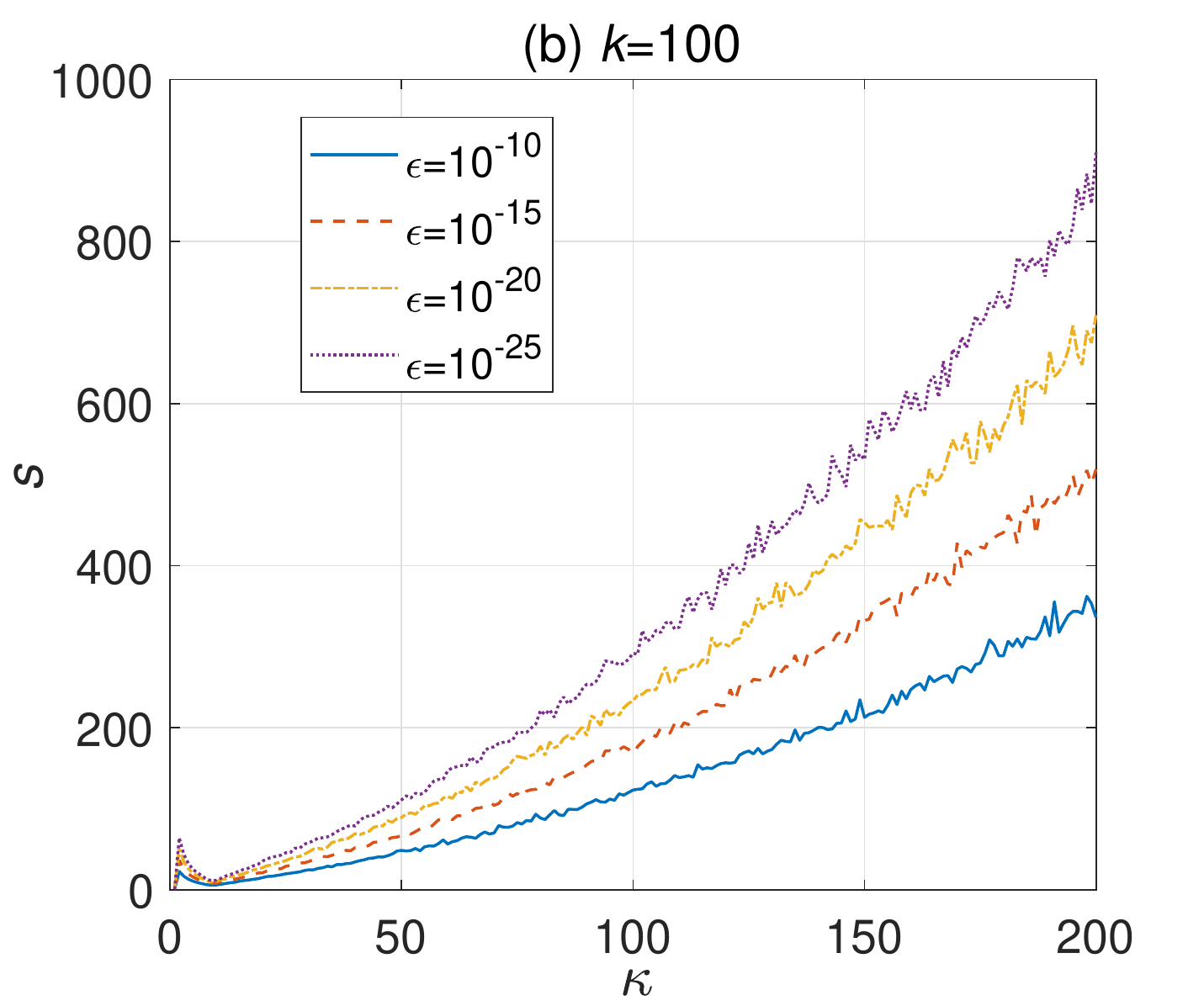}
\end{minipage}
\caption{The relationship of the condition number $\kappa$ and the number of iterations $s$ when
$\epsilon=10^{-20}$ (a) and $k=100$ (b) respectively.}
\label{fig:kappa}
\end{figure*}
\section{Estimate the parameter $c^{(i)}$ and analyze the relative error of $c^{(i)}$ and $\beta_j^{(i)}$}
\label{app:D}
In this appendix, we analyze the value of parameter $c^{(i)}$ and $\theta$ first, and then analyze the relative error of $c^{(i)}$ and $\beta_j^{(i)}$ of our algorithm.

Note that
\begin{eqnarray}
\begin{split}
c^{(i)}\beta_j^{(i)} = \frac{(\sigma_j \beta_j^{(i-1)})^2+ \lambda_2}{\sigma_j^2 \beta_j^{(i-1)}}=\beta_j^{(i-1)} + \frac{\lambda_2}{\sigma_j^2 \beta_j^{(i-1)}}.
\end{split}
\end{eqnarray}
According to \emph{Theorem} \ref{the:2}, we have
\begin{eqnarray}
\begin{split}
c^{(i)}\beta_j^{(i)}
& \ge \min_j c^{(i)} \beta_j^{(i)} \\
& =\min_j \left(\beta_j^{(i-1)} +\frac{\lambda_2} {\sigma_j^2 \beta_j^{(i-1)}}\right) \\
& = \min_j\beta_j^{(i-1)} +\frac{\lambda_2} { \min_j\beta_j^{(i-1)}}.
\end{split}
\end{eqnarray}
Similarly, we have $c^{(i)}\beta_j^{(i)} \le \max_j\beta_j^{(i-1)} +\frac{\lambda_2 \kappa^2} { \max_j\beta_j^{(i-1)}}$.
Because $c > c^{(i)}\beta_j^{(i)}$, we can choose $c \ge \max_j\beta_j^{(i-1)} +\frac{\lambda_2 \kappa^2} { \max_j\beta_j^{(i-1)}}$. Since
\begin{eqnarray}
\max_j\beta_j^{(i-1)} +\frac{\lambda_2 \kappa^2} { \max_j\beta_j^{(i-1)}} \le 1+\lambda_2 \sqrt{k} \kappa^2,
\end{eqnarray}
we can take $c=1+\lambda_2 \sqrt{k} \kappa^2 = O(\sqrt{k}\kappa^2)$.

Based on equation (\ref{eq:B1}),
\begin{eqnarray}
\begin{split}
(c^{(i)})^2 &= (c^{(i)})^2 \sum_j(\beta_j^{(i)})^2 \\
&=\sum_j(\beta_j^{(i-1)} + \frac{\lambda_2}{\sigma_j^2 \beta_j^{(i-1)}})^2 \\
&=\sum_j\left((\beta_j^{(i-1)})^2 + \frac {2 \lambda_2} {\sigma_j^2} + (\frac {\lambda_2} {\sigma_j^2 \beta_j^{(i-1)}})^2 \right)\\
& \ge 1+2k\lambda_2+\lambda_2^2 k^2 =\Omega(k^2),
\end{split}
\end{eqnarray}

thus $c^{(i)}=\Omega(k)$. Then according to Eq (\ref{eq:21}),
\begin{eqnarray}
\begin{split}
\sin(\theta)= \frac{c^{(i)}}{c\sqrt{k}} = \Omega(\frac{k}{k\kappa^2})=\Omega(\frac{1}{\kappa^2}).
\end{split}
\end{eqnarray}

\begin{figure*}[htb]
\centering
\includegraphics[height=8.6cm, width=1\textwidth]{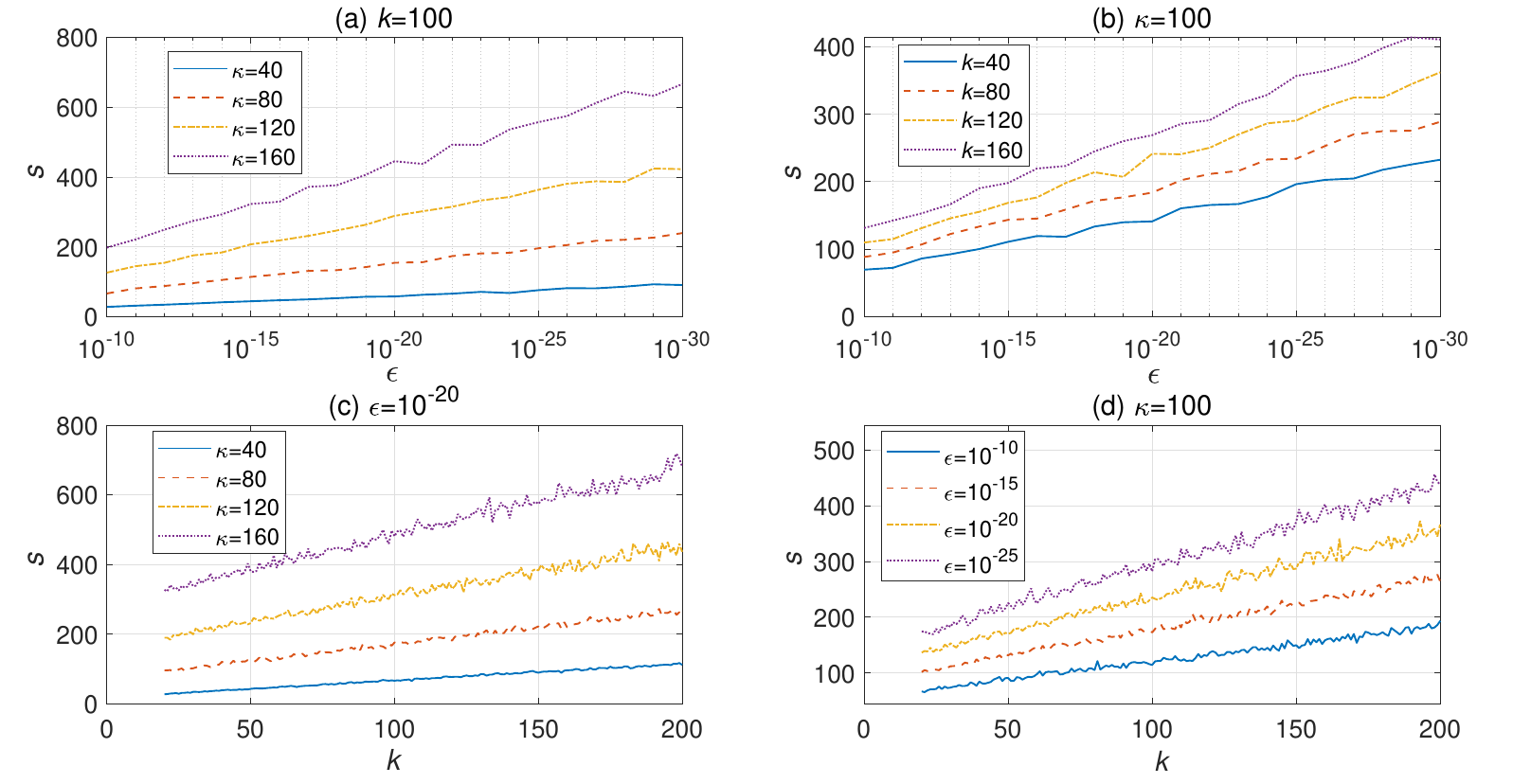}
\caption{Evolution of the number of iterations $s$ respect to the parameter $\epsilon$ ((a) and (b)) and $k$ ((c) and (d)).}
\label{fig:err_k}
\end{figure*}

The relative error of $c^{(i)}$ is
\begin{eqnarray}
\begin{split}
\tilde{\epsilon}_{c} &= \left|\frac{c^{(i)}-{\tilde{c}^{(i)}}}{c^{(i)}}\right|
= \left|\frac{c\sqrt{k}\sin(\theta)-{c\sqrt{k}\sin(\tilde{\theta}})}{c\sqrt{k}\sin(\theta)}\right| \\
&\le \left|\frac{\Delta\theta}{\sin(\theta)}\right| = O(\kappa^2\epsilon_2).
\end{split}
\end{eqnarray}
Thus the relative error of $\beta_j^{(i)}$ of step 4 is
    \begin{eqnarray}
    \begin{split}
        &\left|\frac{\beta_j^{(i)}-{\tilde{\beta}_j^{(i)}}}{\beta_j^{(i)}}\right|
          = \left|\frac{\frac{(\sigma_j \beta_j^{(i-1)})^2+ \lambda_2}{c^{(i)}\sigma_j^2
          \beta_j^{(i-1)}}-\frac{(\widetilde{\sigma}_j \widetilde{\beta}_j^{(i-1)})^2+
          \lambda_2}{\widetilde{c}^{(i)}\widetilde{\sigma}_j^2 \widetilde{\beta}_j^{(i-1)}}}{\frac{(\sigma_j \beta_j^{(i-1)})^2+ \lambda_2}{c^{(i)}\sigma_j^2 \beta_j^{(i-1)}}} \right| \\
         &= O(\left|\frac{\sigma_j^2\lambda_2 + 2\sigma_j \beta_j^{(i-1)} \lambda_2 - \sigma_j^4 (\beta_j^{(i-1)})^2}{\sigma_j^4 (\beta_j^{(i-1)})^3 + \sigma_j (\beta_j^{(i-1)})\lambda_2}\right|\epsilon_1+\tilde{\epsilon}_c) \\
         &= O(\frac{\sigma_j\lambda_2 + 2\beta_j^{(i-1)} \lambda_2}{\sigma_j (\beta_j^{(i-1)})\lambda_2}\epsilon_1+\tilde{\epsilon}_c) \\
         &= O(\kappa^2\sqrt{k}\epsilon_1+\tilde{\epsilon}_c),
    \end{split}
    \end{eqnarray}
where we neglect the quadratic terms of $\epsilon_1$ in the second equation, $\epsilon_1$ is the absolute error of $\beta_j$ and $\sigma_j$.

\section{The number of iterations of the reformulated AOP algorithm}
\label{app:iteration}
In this appendix, we evaluate the number of iterations $s$ of the reformulated AOP algorithm through numerical experiments on randomly generated datasets respect to three parameters, i.e., the condition number $\kappa$ of $\widetilde{X}$,  the precision $\epsilon$, and the number of the principal components $k$ of $\widetilde{X}$ (or the dimensionality of the reduced feature space). By analyzing the steps of the algorithm, we find that the parameters $n$ and $m$ (the dimensionality of $\widetilde{X}$) won't influence the number of iteration directly.

We evaluate the relationship of $s$ and $\kappa$ first. Note that the first step of the reformulated AOP algorithm is computing the PCA of the data matrix $\widetilde{X}$ and then selecting the $k$ largest eigenvalues (i.e. the $\sigma_0^2, \sigma_1^2,...,\sigma_{k-1}^2$) to do the later steps. To reduce the running time, We randomly generate $k$ positive numbers as the $k$ largest eigenvalues of $\widetilde{X}$ to remove the process of PCA. The experimental results is shown in FIG. \ref{fig:kappa}. We set $\epsilon=10^{-20}$ and $k=\{40,80,120,160\}$ in FIG. \ref{fig:kappa}(a), and set $k=100$ and $\epsilon=\{10^{-10}, 10^{-15}, 10^{-20}, 10^{-25}\}$ in FIG. \ref{fig:kappa}(b). The experimental results show that $s$ may have a superlinear dependence on $\kappa$ when $\kappa \ge 20$. Note that the exponential speedups of the two quantum algorithms are based on $\kappa=O(\mathrm{polylog}(mn))$ and the quantum algorithms have advantages when $n$ and $m$ are large. Thus the situation that the value of $\kappa$ is large, i.e. $\kappa \ge 20$, deserves more attention.

The relationship of $\epsilon$ and $s$ is shown in FIG. \ref{fig:err_k}(a) and (b), and the relationship of $k$ and $s$ is shown in FIG. \ref{fig:err_k}(c) and (d). The experimental results in FIG. \ref{fig:err_k} show that $s$ may be linearly dependent on $\log_2(1/\epsilon)$ and $k$. Put all the three parameters together, $s=\Omega(\kappa+\log_2(1/\epsilon)+k)$ may hold.
\\
\\
\bibliography{refe}

\end{document}